\documentclass[pra,aps,showpacs,onecolumn,twoside,superscriptaddress]{revtex4}

\usepackage{amsmath,amsfonts,amssymb,caption,color,epsfig,graphics,graphicx,hyperref,latexsym,mathrsfs,revsymb,theorem,url,verbatim,epstopdf,mathtools,enumerate}

\newtheorem{definition}{Definition}
\newtheorem{proposition}[definition]{Proposition}
\newtheorem{lemma}[definition]{Lemma}

\newtheorem{theorem}[definition]{Theorem}
\newtheorem{corollary}[definition]{Corollary}
\newtheorem{conjecture}[definition]{Conjecture}

\newtheorem{remark}[definition]{Remark}
\newtheorem{example}[definition]{Example}
\newtheorem{question}[definition]{Question}
\newtheorem{memo}[definition]{Memo}

\def\squareforqed{\hbox{\rlap{$\sqcap$}$\sqcup$}}
\def\qed{\ifmmode\squareforqed\else{\unskip\nobreak\hfil
\penalty50\hskip1em\null\nobreak\hfil\squareforqed
\parfillskip=0pt\finalhyphendemerits=0\endgraf}\fi}
\def\endenv{\ifmmode\;\else{\unskip\nobreak\hfil
\penalty50\hskip1em\null\nobreak\hfil\;
\parfillskip=0pt\finalhyphendemerits=0\endgraf}\fi}
\newenvironment{proof}{\noindent \textbf{{Proof.~} }}{\qed}
\def\Dbar{\leavevmode\lower.6ex\hbox to 0pt
{\hskip-.23ex\accent"16\hss}D}
\makeatletter
\def\url@leostyle{%
  \@ifundefined{selectfont}{\def\UrlFont{\sf}}{\def\UrlFont{\small\ttfamily}}}
\makeatother
\def\bcj{\begin{conjecture}}
\def\ecj{\end{conjecture}}
\def\bcr{\begin{corollary}}
\def\ecr{\end{corollary}}
\def\bd{\begin{definition}}
\def\ed{\end{definition}}
\def\bea{\begin{eqnarray}}
\def\eea{\end{eqnarray}}
\def\bem{\begin{enumerate}}
\def\eem{\end{enumerate}}
\def\bex{\begin{example}}
\def\eex{\end{example}}
\def\bim{\begin{itemize}}
\def\eim{\end{itemize}}
\def\bl{\begin{lemma}}
\def\el{\end{lemma}}
\def\bma{\begin{bmatrix}}
\def\ema{\end{bmatrix}}
\def\bpf{\begin{proof}}
\def\epf{\end{proof}}
\def\bpp{\begin{proposition}}
\def\epp{\end{proposition}}
\def\bqu{\begin{question}}
\def\equ{\end{question}}
\def\br{\begin{remark}}
\def\er{\end{remark}}
\def\bt{\begin{theorem}}
\def\et{\end{theorem}}
\def\bmm{\begin{memo}}
\def\emm{\end{memo}}

\def\btb{\begin{tabular}}
\def\etb{\end{tabular}}

\newcommand{\nc}{\newcommand}

\def\g{\gamma}

\def\r{\rho}
\def\s{\sigma}

\def\ps{\psi}

\def\G{\Gamma}

\def\L{\Lambda}

\def\Ps{\Psi}

 \nc{\bbA}{\mathbb{A}} \nc{\bbB}{\mathbb{B}} \nc{\bbC}{\mathbb{C}}
 \nc{\bbD}{\mathbb{D}} \nc{\bbE}{\mathbb{E}} \nc{\bbF}{\mathbb{F}}
 \nc{\bbG}{\mathbb{G}} \nc{\bbH}{\mathbb{H}} \nc{\bbI}{\mathbb{I}}
 \nc{\bbJ}{\mathbb{J}} \nc{\bbK}{\mathbb{K}} \nc{\bbL}{\mathbb{L}}
 \nc{\bbM}{\mathbb{M}} \nc{\bbN}{\mathbb{N}} \nc{\bbO}{\mathbb{O}}
 \nc{\bbP}{\mathbb{P}} \nc{\bbQ}{\mathbb{Q}} \nc{\bbR}{\mathbb{R}}
 \nc{\bbS}{\mathbb{S}} \nc{\bbT}{\mathbb{T}} \nc{\bbU}{\mathbb{U}}
 \nc{\bbV}{\mathbb{V}} \nc{\bbW}{\mathbb{W}} \nc{\bbX}{\mathbb{X}}
 \nc{\bbZ}{\mathbb{Z}}

 \nc{\bA}{{\bf A}} \nc{\bB}{{\bf B}} \nc{\bC}{{\bf C}}
 \nc{\bD}{{\bf D}} \nc{\bE}{{\bf E}} \nc{\bF}{{\bf F}}
 \nc{\bG}{{\bf G}} \nc{\bH}{{\bf H}} \nc{\bI}{{\bf I}}
 \nc{\bJ}{{\bf J}} \nc{\bK}{{\bf K}} \nc{\bL}{{\bf L}}
 \nc{\bM}{{\bf M}} \nc{\bN}{{\bf N}} \nc{\bO}{{\bf O}}
 \nc{\bP}{{\bf P}} \nc{\bQ}{{\bf Q}} \nc{\bR}{{\bf R}}
 \nc{\bS}{{\bf S}} \nc{\bT}{{\bf T}} \nc{\bU}{{\bf U}}
 \nc{\bV}{{\bf V}} \nc{\bW}{{\bf W}} \nc{\bX}{{\bf X}}
 \nc{\bZ}{{\bf Z}}

\nc{\cA}{{\cal A}} \nc{\cB}{{\cal B}} \nc{\cC}{{\cal C}}
\nc{\cD}{{\cal D}} \nc{\cE}{{\cal E}} \nc{\cF}{{\cal F}}
\nc{\cG}{{\cal G}} \nc{\cH}{{\cal H}} \nc{\cI}{{\cal I}}
\nc{\cJ}{{\cal J}} \nc{\cK}{{\cal K}} \nc{\cL}{{\cal L}}
\nc{\cM}{{\cal M}} \nc{\cN}{{\cal N}} \nc{\cO}{{\cal O}}
\nc{\cP}{{\cal P}} \nc{\cQ}{{\cal Q}} \nc{\cR}{{\cal R}}
\nc{\cS}{{\cal S}} \nc{\cT}{{\cal T}} \nc{\cU}{{\cal U}}
\nc{\cV}{{\cal V}} \nc{\cW}{{\cal W}} \nc{\cX}{{\cal X}}
\nc{\cZ}{{\cal Z}}

\nc{\hA}{{\hat{A}}} \nc{\hB}{{\hat{B}}} \nc{\hC}{{\hat{C}}}
\nc{\hD}{{\hat{D}}} \nc{\hE}{{\hat{E}}} \nc{\hF}{{\hat{F}}}
\nc{\hG}{{\hat{G}}} \nc{\hH}{{\hat{H}}} \nc{\hI}{{\hat{I}}}
\nc{\hJ}{{\hat{J}}} \nc{\hK}{{\hat{K}}} \nc{\hL}{{\hat{L}}}
\nc{\hM}{{\hat{M}}} \nc{\hN}{{\hat{N}}} \nc{\hO}{{\hat{O}}}
\nc{\hP}{{\hat{P}}} \nc{\hR}{{\hat{R}}} \nc{\hS}{{\hat{S}}}
\nc{\hT}{{\hat{T}}} \nc{\hU}{{\hat{U}}} \nc{\hV}{{\hat{V}}}
\nc{\hW}{{\hat{W}}} \nc{\hX}{{\hat{X}}} \nc{\hZ}{{\hat{Z}}}

\nc{\hn}{{\hat{n}}}

\def\diag{\mathop{\rm diag}}

\def\lin{\mathop{\rm span}}

\def\sn{\mathop{\rm SN}}

\def\dg{\dagger}

\def\ox{\otimes}

\def\ra{\rightarrow}

\newcommand{\bra}[1]{\langle#1|}
\newcommand{\ket}[1]{|#1\rangle}
\newcommand{\proj}[1]{| #1\rangle\!\langle #1 |}
\newcommand{\ketbra}[2]{|#1\rangle\!\langle#2|}

\def\Dbar{\leavevmode\lower.6ex\hbox to 0pt
{\hskip-.23ex\accent"16\hss}D}

\begin{document}
\title{On the PPT Square Conjecture for $n=3$}

\date{\today}


\author{Lin Chen}
\email{linchen@buaa.edu.cn}
\affiliation{School of Mathematics and Systems Science, Beihang University, Beijing 100191, China}
\affiliation{International Research Institute for Multidisciplinary Science, Beihang University, Beijing 100191, China}

\author{Yu Yang}
\email{yy19900320@icloud.com}
\affiliation{Department of Mathematics and Statistics, Chongqing Technology and Business University, Chongqing 400067, China}

\author{Wai-Shing Tang}
\email{mattws@nus.edu.sg}
\affiliation{Department of Mathematics, National University of Singapore, 10 Lower Kent Ridge Road, Singapore 119076, Republic of Singapore}

\begin{abstract}
We present the PPT square conjecture introduced by M. Christandl. We prove the conjecture in the case $n=3$ as a consequence of the fact that two-qutrit PPT states have Schmidt at most two. The PPT square conjecture in the case $n\ge4$ is still open. We present an example to support the conjecture for $n=4$. 
\end{abstract}

\maketitle

\section{Introduction}
\label{sec:int}


In quantum physics, quantum operations are implemented by quantum channels. The quantum channel is a completely positive trace preserving (CPTP) map between two matrix algebras \cite{choi1972}. The composition of quantum channels is a fundamental operation in quantum information, and we need quantum channels creating entangled states useful for quantum-information tasks. Nevertheless, some channels create positive partial transpose (PPT) states, and they become little useful \cite{bchw2015,hhho2005,lg2015}. 
Such channels are called 
PPT channels. It has been further conjectured by M. Christandl that the compositions of PPT channels is an entanglement-breaking channel  \cite{Christhesis,cppt}. That is, the Choi matrix of the composition is a separable state \cite{hsr2003}. The conjecture is known now as the PPT square conjecture/ Proving it would imply a deeper understanding of the difference between PPT and separable states, as well as a novel method for attaching the separability problem.


The conjecture has received a lot of attentions recently. First
the conjecture holds asymptotically when the distance between the iterates of any unital or trace-preserving PPT channel and the set of entanglement breaking maps tends to zero \cite{kmp2017}. Next every unital PPT channel becomes entanglement breaking after a finite number of iterations \cite{rjp2017}. In the finite dimensional case, very recently it has been proved that the conjecture in dimension three and gave some examples to support the conjecture in dimension four \cite{chw2018}. As the first main result of this paper, we independently show the same result by using a novel way of deciding separable states in $M_3\ox M_3$\footnote{We have claimed the independence of our result from \cite{chw2018} by private communication with Alexander Muller Hermes.}. This is presented in Theorem \ref{thm:n=3}. Different from  \cite{chw2018}, our method can be extended to decide the separability of some $n\times 3$ PPT states in Theorem \ref{thm:n=3,stronger}. This is the second main result of this paper. As far as we know, this is the latest progress on this long-standing open problem. To further investigate the PPT square conjecture for high dimensions, we construct a $4\times4$ PPT entangled state, extract its PPT map and show that the Choi matrix of composition of the PPT map is a separable state. So it supports the PPT square conjecture. We further discuss the example, and construct a more general family of maps satisfying Conjecture \ref{cj:pps} in Theorem \ref{thm:phi=sumj}.

The rest of this paper is organized as follows. In Sec. \ref{sec:3x3}, we construct the definitions we use in this paper. We present the main problem as Conjecture \ref{cj:pps}, and prove the special case on two-qutrit states. We further generalize our results to $n\times3$ states in Corollary \ref{cr:nx3}. In Sec. \ref{sec:4x4}, we construct an example of $4\times4$ PPT state to support Conjecture \ref{cj:pps}, and further investigate the example in Sec. \ref{sec:dis4x4}. Finally we conclude in Sec. \ref{sec:con}.

\section{Proving the PPT square conjecture for $n=3$}
\label{sec:3x3}

We shall work with bipartite states on the space $\cH_A\ox\cH_B$. Linear maps that are both completely positive and completely copositive are called PPT binding maps. 

Let us consider the composition $\phi_2\circ\phi_1$ of two PPT maps $\phi_1$ and $\phi_2$ where $\phi_{1},\phi_2\in M_n(\mathbb{C})\otimes M_n(\mathbb{C})$. Let $M_{m,n}$ be the set of $m\times n$ complex matrices and $B(M_{m,n})$ be the set of all linear maps on $M_{m,n}(\mathbb{C})$. If $m=n$ then we denote $M_n:=M_{n,m}$. We have the Kraus decomposition for a completely positive and trace preserving map $\L$, that is, $\L(*)=\sum_i A_i(*)A_i^\dg$ and $\sum_i A_i^\dg A_i=I$. 
\bd
Let  $C_{\phi_1\circ\phi_2}$ be the Choi matrix of the composition of two CPTP maps $\phi_1$ and $\phi_2$. We shall call $C_{\phi_1\circ\phi_2}$ the composition of two Choi matrices $C_{\phi_1}$ and $C_{\phi_2}$.
\ed
All PPT maps on $M_2(\mathbb{C})$ are separable due to the Peres-Horodecki separation criterion \cite{peres1996,woronowicz1976}. So the first non-trivial case lies in $B(M_3(\mathbb{C}))$ and we shall always confine ourselves the $n\geq3$ cases. Obviously, the composition of two PPT maps is still PPT. The following conjecture is referred to as the PPT square conjecture. It is known that the two conjectures in Conjecture \ref{cj:pps} are equivalent. We refer readers to recent progress on the conjecture in \cite{Collins2018}. 
\bcj
\label{cj:pps}
(i) If $\phi$ is a PPT map on $M_n$, then $\phi\circ\phi$ is separable.
\\
(ii) If $\phi_1$ and $\phi_2$ are PPT maps on $M_n$, then $\phi_1\circ\phi_2$ is separable.
\ecj

The following result proves a special case of Conjecture \ref{cj:pps} (i). This is the first main result of this paper. 
\begin{theorem}
\label{thm:n=3}
Conjecture \ref{cj:pps} (i) holds for $n=3$.	
\end{theorem}
\begin{proof}
Suppose $\r$ is an arbitrary quantum state in $M_3\ox M_3$. So $\s:=(I_3\ox\phi)(\r)$ is a PPT state in $M_3\ox M_3$. It is known that $\s$ has Schmidt number at most two \cite{ylt16}. Let $\s=\sum_j p_j\proj{a_j}$ where each $\ket{a_j}$ has Schmidt rank at most two. 
That is, $\ket{a_j}\in\cK_j\simeq \bbC^3\ox\bbC^2$. So each state $(I_3\ox\phi)(\proj{a_j})$ is a PPT state in $M_3\ox M_2$ up to local equivalence. The Peres-Horodecki criterion says that $(I_3\ox\phi)(\proj{a_j})$ is separable. Using the convex sum of $\s$ we obtain that $(I_3\ox\phi)(\s)=(I_3\ox(\phi\circ\phi))(\r)$ is separable. Choosing $\r$ as the maximally entangled state implies the assertion.
\end{proof}

The following is a corollary of Theorem \ref{thm:n=3}. It provides a novel method of deciding separable states in $M_n\ox M_3$.

\begin{corollary}
\label{cr:nx3}
If $\phi_1,\phi_2$ are two PPT maps on $M_3$ and $\r\in M_n\ox M_3$ then the state $(I_n\ox (\phi_1\circ\phi_2))\r$ is separable. 	
\end{corollary}
\begin{proof}
It suffices to prove the assertion when $\r_A$ is the maximally mixed state, namely $\r_A={1\over n}I_n$. We assume $\r=\sum_j \proj{\ps_j}$, where $\ket{\ps_j}=(A_j\ox I_3)\ket{\Ps_3}$, where the isometry $A_j:\bbC^3\ra\bbC^n$ and $\ket{\Ps_3}={1\over\sqrt3}\sum^2_{i=0}\ket{ii}$ is the two-qutrit maximally entangled state. 
Using Theorem \ref{thm:n=3} and the equivalence of the two statements in Conjecture \ref{cj:pps}, we obtain that the state $(I_n\ox (\phi\circ\phi))(\proj{\ps_j}$) is separable for any $j$. So the state $(I_n\ox (\phi_1\circ\phi_2))\r$ is separable. This completes the proof.	
\end{proof}

Note that the proofs of Theorem \ref{thm:n=3} and Corollary \ref{cr:nx3} both apply to the case when $\phi,\phi_1,\phi_2$ are not trace-preserving.
Further, the proof of Theorem \ref{thm:n=3} does not rely on the fact that $(I_3\otimes\phi)(\r)$ is a PPT. In fact, it relies on the fact that $(I_3\otimes\phi)(\r)$ has Schmidt number at most two\cite{ylt16}. So we have the following result. This is the second main result of this paper. 
\begin{theorem}
\label{thm:n=3,stronger}	
Suppose $\phi_1$ is a PPT map on $M_3$, $\phi_2$ is a completely positive map on $M_3$, and the bipartite state $\r\in M_n\otimes M_3$. If $(I_n\otimes\phi_2)\r$ has Schmidt number at most two then the state $(I_n\ox (\phi_1\circ\phi_2))\r$ is separable. 	 
\end{theorem}

Corollary \ref{cr:nx3} and Theorem \ref{thm:n=3,stronger} provide two channels $\phi_1,\phi_2$ such that their combination becomes both entanglement-breaking channels. If $\phi_2$ is a PPT map then $(I_n\otimes\phi_2)\r$ has Schmidt number at most two. So Theorem \ref{thm:n=3,stronger} is stronger than Corollary \ref{cr:nx3}.

The next case for studying Conjecture \ref{cj:pps} is $M_4\otimes M_4$. We propose the following conjecture. It is not included in Theorem \ref{thm:n=3,stronger}, because the latter discusses only states in $M_n\otimes M_3$.

\begin{conjecture}
\label{cj:4x4}
If $\phi$ is a PPT map on $M_4$ and $\r\in M_4\ox M_4$ has Schmidt rank two, then the state $(I_4\ox (\phi\circ\phi))\r$ is separable. 
\end{conjecture}
In the next section we construct an example supporting Conjecture \ref{cj:4x4} and thus Conjecture \ref{cj:pps}.

\section{An example to support the PPT square conjecture for $n=4$}
\label{sec:4x4}

It is unknown whether Conjecture \ref{cj:pps} for $n\ge4$ is true. 
Some examples satisfying the conjecture have been constructed in \cite{Collins2018}. They respectively rely on randomness, graphs, certain symmetries, or are Gaussian channels. Different from these examples, in this section we shall construct an example satisfying Conjecture \ref{cj:pps} for $n=4$. In particular, we construct a $4\times4$ PPT entangled state $(I_4\ox\phi)(\proj{\ps})$ in \eqref{eq:sigma} with Kraus operators in \eqref{eq:kraus}, such that the state
$(I_4\otimes(\phi\circ\phi))(\proj{\ps})$ defined via \eqref{eq:gamma} turns out to be separable. Here $\phi$ is the PPT map on $\cH_B$ defined in \eqref{eq:sigma} and $\ket{\ps}=\ket{00}+\ket{11}+\ket{22}+\ket{33}$ is the $4\times4$ maximally entangled state.


We construct the following state inspired by \cite[Sec VII.B]{dps2004}.
\begin{eqnarray}
\r_2
\label{eq:sn2-1}
&=&(\ket{00}+\ket{11}+\ket{22})
(\bra{00}+\bra{11}+\bra{22})	
+
\proj{02}+\proj{20}+
\\\label{eq:sn2-2}
&+&
(\ket{01}+\ket{10}+\ket{33})
(\bra{01}+\bra{10}+\bra{33})	
+
\proj{03}+\proj{31}+
\\\label{eq:sn2-3}
&+&
\proj{12}+\proj{13}+
\proj{30}+\proj{21}.
\end{eqnarray}
The partial transpose of $\r_2$ is
\begin{eqnarray}
\r_2^{\G}&=&
(
\proj{00}+\proj{11}+\proj{22}
+\ketbra{01}{10}+\ketbra{10}{01}
+\ketbra{02}{20}+\ketbra{20}{02}
+\ketbra{12}{21}+\ketbra{21}{12}
)
\\&+&
(
\proj{01}+\proj{10}+\proj{33}
+\ketbra{00}{11}+\ketbra{11}{00}
+\ketbra{31}{03}+\ketbra{03}{31}
+\ketbra{13}{30}+\ketbra{30}{13}
)
\\&+&	
\proj{02}+\proj{20}+
\proj{03}+\proj{31}+
\proj{12}+\proj{13}
+\proj{30}+\proj{21}.
\end{eqnarray}
Let $\r_2\ra \s:={1\over3}(\diag(a,b,c,1)\ox I_4)\r_2(\diag(a,b,c,1)\ox I_4)$ with positive $a,b,c$. Then
\begin{eqnarray}
\s
&=&
{1\over3}
\bigg(
(a\ket{00}+b\ket{11}+c\ket{22})
(a\bra{00}+b\bra{11}+c\bra{22})	
\\&+&
(a\ket{01}+b\ket{10}+\ket{33})
(a\bra{01}+b\bra{10}+\bra{33})	
\\&+&
a^2\proj{02}+a^2\proj{03}+b^2\proj{12}+b^2\proj{13}+
c^2\proj{20}+\proj{30}+c^2\proj{21}+\proj{31}
\bigg).
\end{eqnarray}
To make $\phi$ as a quantum channel, we require
$
\s_A=
{1\over3}(4a^2\proj{0}+4b^2\proj{1}+3c^2\proj{2}+3\proj{3})=I_4.	
$
So we have $a=b={\sqrt3\over2}$, and $c=1$. We have 
\begin{eqnarray}
\label{eq:sigma}
\s
=
(I_4\otimes\phi)(\proj{\ps})
:=\sum^{10}_{j=1}(I_4\ox P_j)\proj{\ps}(I_4\ox P_j^\dg),	
\end{eqnarray}
where the Kraus operators of map $\phi$ are
\begin{eqnarray}
\label{eq:kraus}
P_1&=&{1\over\sqrt3}\diag(a,b,c,0)
=P_1^\dg,
\quad\quad\quad\quad
\quad\quad\quad\quad
P_2=
{1\over\sqrt3}
(a\ketbra{1}{0}+b\ketbra{0}{1}+\proj{3})
=P_2^\dg,
\notag\\	
P_3&=& {a\over\sqrt3}\ketbra{2}{0}, 
\quad\quad\quad\quad
P_4= {a\over\sqrt3}\ketbra{3}{0}, 
\quad\quad\quad\quad
P_5= {b\over\sqrt3}\ketbra{2}{1}, 
\quad\quad\quad\quad	
P_6= {b\over\sqrt3}\ketbra{3}{1}, 
\notag\\	
P_7&=& {c\over\sqrt3}\ketbra{0}{2}, 
\quad\quad\quad\quad
P_8= {c\over\sqrt3}\ketbra{1}{2}, 
\quad\quad\quad\quad
P_9= {1\over\sqrt3}\ketbra{0}{3}, 
\quad\quad\quad\quad
P_{10}= {1\over\sqrt3}\ketbra{1}{3}, 
\end{eqnarray}
satisfy $\sum_j P_j^\dg P_j=I_4$. So we obtain the PPT map $\phi:\bbC^4\ra\bbC^4$. For our purpose we investigate the PPT map $\phi\circ\phi$ with Kraus operators $\{P_iP_j\}$.
By computing one can show that the nonzero operators in $\{P_iP_j\}$ are the following 44 matrices in $M_4\ox M_4$. 
\begin{eqnarray}
P_1P_1&=&{1\over3}\diag(a^2,b^2,c^2,0),
\quad\quad\quad\quad
\quad\quad\quad\quad
P_1P_2= 
P_2P_1=
{ab\over3}
(\ketbra{1}{0}+\ketbra{0}{1})
,
\\	
P_1P_3&=& {ac\over3}\ketbra{2}{0}, 
\quad\quad\quad\quad	
P_1P_5= {bc\over3}\ketbra{2}{1}, 
\\	
P_1P_7&=& {ac\over3}\ketbra{0}{2}, 
\quad\quad\quad\quad	
P_1P_8= {bc\over3}\ketbra{0}{3}, 
\quad\quad\quad\quad	
P_1P_9= {a\over3}\ketbra{1}{2}, 
\quad\quad\quad\quad	
P_1P_{10}= {b\over3}\ketbra{1}{3}, 
\end{eqnarray}
\begin{eqnarray}
P_2P_2&=& 
{1\over3}
(ab\ketbra{0}{0}+ab\ketbra{1}{1}+\proj{3}),
\\	
P_2P_4&=& {a\over3}\ketbra{3}{0}, 
\quad\quad\quad\quad
P_2P_6= {b\over3}\ketbra{3}{1}, 
\\	
P_2P_7&=& {ac\over3}\ketbra{1}{2}, 
\quad\quad\quad\quad
P_2P_8= {bc\over3}\ketbra{0}{2}, 
\quad\quad\quad\quad
P_2P_9= {a\over3}\ketbra{1}{3}, 
\quad\quad\quad\quad
P_2P_{10}= {b\over3}\ketbra{0}{3}, 
\end{eqnarray}
and
\begin{eqnarray}
P_3P_1&=&{a^2\over3}\ketbra{2}{0},
\quad\quad\quad\quad
P_3P_2=
{ab\over3}
\ketbra{2}{1},	
\quad\quad\quad\quad
P_3P_7= {ac\over3}\ketbra{2}{2}, 
\quad\quad\quad\quad
P_3P_9= {a\over3}\ketbra{2}{3}, 
\\
P_4P_1&=&{a^2\over3}\ketbra{3}{0},
\quad\quad\quad\quad
P_4P_2=
{ab\over3}\ketbra{3}{1},	
\quad\quad\quad\quad
P_4P_7= {ac\over3}\ketbra{3}{2}, 
\quad\quad\quad\quad
P_4P_9= {a\over3}\ketbra{3}{3}, 
\\
P_5P_1&=&{b^2\over3}\ketbra{2}{1},
\quad\quad\quad\quad
P_5P_2=
{ab\over3}
\ketbra{2}{0},	
\quad\quad\quad\quad
P_5P_8= {bc\over3}\ketbra{2}{2}, 
\quad\quad\quad\quad
P_5P_{10}= {b\over3}\ketbra{2}{3}, 
\\
P_6P_1&=&{b^2\over3}\ketbra{3}{1},
\quad\quad\quad\quad
P_6P_2={ab\over3}\ketbra{3}{0},	
\quad\quad\quad\quad
P_6P_8= {bc\over3}\ketbra{3}{2}, 
\quad\quad\quad\quad
P_6P_{10}= {b\over3}\ketbra{3}{3},
\\
P_7P_1&=&{c^2\over3}\ketbra{0}{2},
\quad\quad\quad\quad
P_7P_3= {ac\over3}\ketbra{0}{0}, 
\quad\quad\quad\quad
P_7P_5= {bc\over3}\ketbra{0}{1}, 
\\
P_8P_1&=&{c^2\over3}\ketbra{1}{2},
\quad\quad\quad\quad
P_8P_3= {ac\over3}\ketbra{1}{0}, 
\quad\quad\quad\quad
P_8P_5= {bc\over3}\ketbra{1}{1}, 
\\
P_9P_2&=&{1\over3}\ketbra{0}{3},
\quad\quad\quad\quad
P_9P_4={a\over3}\ketbra{0}{0}, 
\quad\quad\quad\quad
P_9P_6= {b\over3}\ketbra{0}{1}, 
\\
P_{10}P_2&=&{1\over3}\ketbra{1}{3},
\quad\quad\quad\quad
P_{10}P_4={a\over3}\ketbra{1}{0}, 
\quad\quad\quad\quad
P_{10}P_6= {b\over3}\ketbra{1}{1}.
\end{eqnarray}
For convenience, we define the invertible diagonal matrix 
$D=\diag(4,4,3,3)$. We perform the map $\phi$ on the state $\s$ in \eqref{eq:sigma}, and investigate the separability of the resulting state as follows.
\begin{eqnarray}
\label{eq:gamma}
\g &:=&	
(I_4\ox D)
\bigg(
(I_4\otimes(\phi\circ\phi))(\proj{\ps})
\bigg)(I_4\ox D)
\\&=&
(I_4\ox D)
\bigg(
\sum^{10}_{j,k=1}(I_4\ox P_jP_k)\s(I_4\ox P_k^\dg P_j^\dg)
\bigg)(I_4\ox D)
\\&=&	
\bigg(\ket{00}+\ket{11}+\ket{22}\bigg)
\bigg(\bra{00}+\bra{11}+\bra{22}\bigg)
\\&+&
2
(\ket{01}+\ket{10})
(\bra{01}+\bra{10})
+	
\bigg(\ket{00}+\ket{11}+\ket{33}\bigg)
\bigg(\bra{00}+\bra{11}+\bra{33}\bigg)
\\&+&
{8\over3}\proj{00}+
{8\over3}\proj{10}+
{40\over9}\proj{20}+
{40\over9}\proj{30}
+
{11\over3}\proj{01}+
{11\over3}\proj{11}+
{52\over9}\proj{21}+
{52\over9}\proj{31}
\\&+&
{21\over16}\proj{02}+
{21\over16}\proj{12}+
{3\over4}\proj{22}+
{3\over4}\proj{32}
+
{15\over8}\proj{03}+
{15\over8}\proj{13}+
{3\over2}\proj{23}+
{3\over2}\proj{33}.
\end{eqnarray}
The partial transpose of $\g$ is
\begin{eqnarray}
\g^\G &=&		
\proj{00}+\proj{11}+\proj{22}
+\ketbra{01}{10}+\ketbra{10}{01}
+\ketbra{12}{21}+\ketbra{21}{12}
+\ketbra{20}{02}+\ketbra{02}{20}
\\&+&
2\proj{01}+2\proj{10}
+2\ketbra{00}{11}+2\ketbra{11}{00}
\\&+&	
\proj{00}+\proj{11}+\proj{33}
+\ketbra{01}{10}+\ketbra{10}{01}
+\ketbra{03}{30}+\ketbra{30}{03}
+\ketbra{13}{31}+\ketbra{31}{13}
\\&+&
{8\over3}\proj{00}+
{8\over3}\proj{10}+
{40\over9}\proj{20}+
{40\over9}\proj{30}
+
{11\over3}\proj{01}+
{11\over3}\proj{11}+
{52\over9}\proj{21}+
{52\over9}\proj{31}
\\&+&
{21\over16}\proj{02}+
{21\over16}\proj{12}+
{3\over4}\proj{22}+
{3\over4}\proj{32}
+
{15\over8}\proj{03}+
{15\over8}\proj{13}+
{3\over2}\proj{23}+
{3\over2}\proj{33}.
\end{eqnarray}
By splitting $\g^\G$ into the sum of a few two-qubit states $\sum_{i,j}p_{ij}\ketbra{ij}{ji}$ with $p_{ij}=0,1$ or $2$ for $i\ne j$, one can show that each of the two-qubit states has PPT. So they are separable by the Peres-Horodecki criterion. 
Summing up them implies that $\g^\G$ is separable. The definition of $\g$ implies that the state $(I_4\otimes(\phi\circ\phi))(\proj{\ps})$ is also separable.

One can show that $\r_2^\G$ is positive semidefinite. Further, $\r_2$ is locally equivalent to the PPT entangled state in  \cite[Sec VII.B]{Doherty2004Complete}.
So $\r_2$ is a PPT entangled state, and $\sn(\r_2)=\sn(\r_2^\G)=2$. 
In particular $\sn(\r_2)=2$ follows from the fact that the two states in \eqref{eq:sn2-1} and \eqref{eq:sn2-2} both have Schmidt number two. So Conjecture \ref{cj:4x4} holds for our example by choosing $\r_2=(I_4\ox\phi)\r$ in Conjecture \ref{cj:4x4}.

\section{Discussion on the example}
\label{sec:dis4x4}

In this section we investigate the example of last section, and present Theorem \ref{thm:phi=sumj} to cover the example. The state $\r_2$ in \eqref{eq:sn2-1}-\eqref{eq:sn2-3} can be written as $\r_2=\r_3+\r_4$ where the two states
\begin{eqnarray}
\r_3&=&(\ket{00}+\ket{11}+\ket{22})
(\bra{00}+\bra{11}+\bra{22})	
+
\proj{02}+\proj{20}
+
\proj{12}+\proj{21},
\\
\r_4 &=&
(\ket{01}+\ket{10}+\ket{33})
(\bra{01}+\bra{10}+\bra{33})	
+
\proj{03}+\proj{31}
+
\proj{13}+\proj{30}.
\end{eqnarray}
So $\r_3$ and $\r_4$ respectively act on the space $\bbC^3\ox\bbC^3$ and $\lin\{\ket{0},\ket{1},\ket{3}\}\ox\lin\{\ket{0},\ket{1},\ket{3}\}$. That is they are both two-qutrit states. Further, they both have Schmidt rank two from the last section. We have
\begin{eqnarray}
(I_4\ox\phi)\s\sim
(I_4\ox\phi)\r_2	
=&&
\sum^{10}_{j=1}
(I_4\ox P_j)
\r_2
(I_4\ox P_j^\dg)
\\=&&
\sum^{10}_{j=1}
(I_4\ox P_j)
\r_3
(I_4\ox P_j^\dg)
+
\sum^{10}_{j=1}
(I_4\ox P_j)
\r_4
(I_4\ox P_j^\dg)
,
\end{eqnarray}
where the last two sums respectively stand for the direct sum of a two-qutrit states and a product state in terms of the Kraus operators $P_j$'s. So they are both separable from Theorem \ref{thm:n=3,stronger}. From this argument we conclude the following result.
\begin{theorem}
\label{thm:phi=sumj}
Suppose $\phi=\sum_j \phi_j$ such that for any $j$ we have $\phi_j:\bbC^{3\times p}\ra\bbC^{3\times3}$ are both PPT maps. Suppose $\phi'$ is a PPT map such that $(I\otimes\phi')\r$ has Schmidt number two. Then $(I_n\otimes \phi\circ\phi')(\r)$ is separable. 
\end{theorem}



\section{Conclusions}
\label{sec:con}

We have shown that the PPT square conjecture holds for $n=3$ as a consequence of the fact that $3\times3$ PPT states have Schmidt number at most two. That is, every CPTP map in $B(M_3(\mathbb{C}),M_3(\mathbb{C}))$ has index of separability two. (If a linear map is entanglement breaking after finite iterations, the map is said to have a finite index of separability.) We have extended it to a general case for the composition of $\phi_1$ being PPT while $\phi_2$ being CPTP. Further, we have proposed a conjecture as a special case of the PPT conjecture when $n=4$ and the input quantum state $\rho$ is of Schmidt number two. We also have provided a non-trivial concrete example to support the PPT square conjecture when $n=4$. In this case a counterexample is widely believed to exists. The next step for attacking the PPT square conjecture is to investigate more $4\times4$ PPT entangled states by checking their Schmidt number and the relevlant PPT maps.

\section*{Acknowledgements}

LC was supported by the NNSF of China (Grant No. 11871089), Beijing Natural Science Foundation (4173076), and the Fundamental Research Funds for the Central Universities (Grant Nos. KG12040501, ZG216S1810 and ZG226S18C1). 
Yu Yang was supported by Chongqing Technology and Business University Research Fund and Chongqing Key Laboratory of Social Economy and Applied Statistics. Wai-Shing Tang was partially supported by Singapore Ministry of Education Academic Research Fund Tier 1 Grant (No. R-146-000-266-114).









\begin{references}

\bibitem{choi1972}
M. D. Choi, Positive linear maps on $C^*$-algebras, Canad. Math. J., 24:520-529, 1972.

\bibitem{bchw2015}
S. Ba\"uml, M. Christandle, K. Horodecki, and Andreas Winter, Limitations on quantum key repeaters, Nature Communi
cations., 6:6908, 2015.

\bibitem{hhho2005}
Secure key from bound entanglement. Physical Review
Letters., 94(16):160502, 2005.

\bibitem{lg2015}
L. Lami and V. Giovannetti, Entanglementbreaking indices, Journal of Mathematical Physics., 56(9):092201, 2015.

\bibitem{Christhesis}
M. Christandl. Bipartite entanglement: A cryptographic point of view. University of Cambridge., 2005.

\bibitem{cppt}
M. Christandl. Ppt square conjecture. Banff International Research Station workshop: Operator structures in Quantum
Information Theory., 2012.



\bibitem{hsr2003}
M. Horodecki, P. Shor, and M. Ruskai. Entanglement breaking channels. Reviews in Mathematical Physics.,
15(06):629–641, 2003.

\bibitem{kmp2017}
Matthew Kennedy, Nicolas A. Manor, and Vern I. Paulsen. Composition of ppt maps. arXiv:1710.08475, 2017.


\bibitem{rjp2017}
Mizanur Rahaman, Sam Jaques, and Vern I. Paulsen. Eventually entanglement breaking maps. arXiv:1710.08475, 2017.

\bibitem{chw2018}
M. Christandl, M. Hermes and M. Wolf, When Do Composed Maps Become Entanglement Breaking?, arXiv:1807.01266, 2018.

\bibitem{woronowicz1976}
S. L. Woronowicz. Positive maps of low dimensional matrix algebras. Reports on Mathematical Physics., 10(2):165–183,
1976


\bibitem{Collins2018}
B. Collins, Z. Yin and P. Zhong, The PPT square conjecture holds generically for some classes of independent states, arXiv:1803.00143, 2018.

\bibitem{peres1996}
A. Peres, Separability Criterion for Density Matrices, Phys. Rev. Lett. 1413(77), 1996.
	
	
	
\bibitem{ylt16}
Yu Yang, Denny H. Leung and Wai-Shing Tang, All 2-positive linear maps from {$M_3(\mathbb{C})$} to {$M_3(\mathbb{C})$} are decomposable, Linear Algebra and its Applications 233-247 (503), 2016.

\bibitem{hermes}
It was claimed by M\"uller-Hermes in the open problem section of the trimester program of the
Centre Emile Borel ''Analysis in Quantum Information Theory''.

\bibitem{dps2004}
Andrew C. Doherty, Pablo A. Parrilo, and Federico M. Spedalieri. Complete family of separability criteria. Physical Review A, 69(2):022308, 2004.
\end{references}
\end{document}